\documentclass[pdflatex,sn-mathphys-num]{sn-jnl}

\usepackage{graphicx}
\usepackage{multirow}
\usepackage{amsmath,amssymb,amsfonts}
\usepackage{amsthm}
\usepackage{mathrsfs}
\usepackage[title]{appendix}
\usepackage{xcolor}
\usepackage{textcomp}
\usepackage{manyfoot}
\usepackage{booktabs}
\usepackage{algorithm}
\usepackage{algorithmicx}
\usepackage{algpseudocode}
\usepackage{listings}

\theoremstyle{thmstyleone}
\newtheorem{theorem}{Theorem}
\newtheorem{lemma}[theorem]{Lemma}
\newtheorem{proposition}[theorem]{Proposition}
\newtheorem{corollary}[theorem]{Corollary}

\theoremstyle{thmstyletwo}

\newtheorem{remark}{Remark}

\theoremstyle{thmstylethree}
\newtheorem{definition}{Definition}
\newtheorem{assumption}[theorem]{Assumption}

\begin{document}

\title[Quantum Adiabatic Theorem for Non-Hermitian Dynamics]{The Quantum Adiabatic Theorem for Non-Hermitian Dynamics}

\author[1]{\fnm{G. M. V. S.} \sur{Aponso}}\email{venura.sachi@gmail.com}

\author*[1]{\fnm{B. P. W.} \sur{Fernando}}\email{bpw@sci.sjp.ac.lk}

\author[1]{\fnm{K. K. W. A. S.} \sur{Kumara}}\email{sarath@sjp.ac.lk}

\author[2]{\fnm{R. P. K.} \sur{De Silva}}\email{kaushika.desilva@uwa.edu.au}

\affil[1]{\orgdiv{Department of Mathematics}, \orgname{University of Sri Jayewardenepura}, \country{Sri Lanka}}

\affil[2]{\orgdiv{Centre for Quantum Information, Simulation and Algorithms}, \orgname{University of Western Australia}, \country{Australia}}

\abstract{We establish a quantitative adiabatic estimate for a class of finite-dimensional non-Hermitian Schr\"odinger dynamics. The original non-Hermitian Hamiltonian is assumed to be diagonalizable with real spectrum and non-crossing eigenvalues. We then construct a dynamically compatible time-dependent metric operator, and its positive square root defines a Dyson map. The associated Dyson-transformed Hamiltonian is Hermitian. 

When this Dyson-transformed Hamiltonian satisfies the Hermitian adiabatic assumption, the standard resolvent-projection method gives an explicit adiabatic estimate in the Hermitian representation. Pulling this estimate back through the Dyson map gives an approximation in the original non-Hermitian representation. The resulting Dyson-pulled-back projections are then compared with the spectral projections of the original non-Hermitian Hamiltonian by using a contour-resolvent estimate. The final bound contains two contributions: the pulled-back Hermitian adiabatic error and the projection-comparison error. A two-level non-Hermitian model is presented to illustrate the hypotheses of the
theorem and the two contributions appearing in the final estimate.}



\maketitle

\section{Introduction}\label{sec1}

The quantum adiabatic theorem is a fundamental result in quantum mechanics, describing the evolution of systems governed by slowly varying Hamiltonians. In its operator-theoretic form, established by Kato~\cite{Kato1950}, the theorem asserts that if a Hermitian Hamiltonian depends smoothly on time and possesses a spectral gap, then a system initially prepared in a given eigenspace remains close to the corresponding instantaneous eigenspace throughout the evolution, with an error of order $O(1/T)$. This framework has been refined using resolvent-projection methods, leading to explicit error bounds in terms of spectral gaps and regularity assumptions~\cite{Avron1987,Jansen2007,Scherer2019}. These results form the basis of many applications, including geometric phases, molecular dynamics, and adiabatic quantum computation.

A key structural feature underlying the classical adiabatic theorem is the Hermiticity of the Hamiltonian. This guarantees unitary evolution and allows the use of orthogonal spectral projections and resolvent-based techniques. In contrast, for non-Hermitian Hamiltonians, the evolution is not unitary with respect to the standard inner product, and the spectral structure is not preserved in a direct way. As a result, the standard adiabatic framework cannot be applied without modification.

Several extensions of the adiabatic principle have been developed in settings that go beyond the standard closed Hermitian case. In the context of open quantum systems, Sarandy and Lidar~\cite{SarandyLidar2005} formulated an adiabatic approximation for time-dependent master equations by working at the level of the dynamical superoperator and its Jordan blocks. In a different direction, Fleischer and Moiseyev~\cite{FleischerMoiseyev2005} derived an adiabatic theorem for non-Hermitian time-dependent open systems using the $(t,t')$-formalism together with complex scaling. Earlier work by Mostafazadeh~\cite{Mostafazadeh1999} also showed that adiabatic structures and geometric phases continue to arise for non-Hermitian Hamiltonians. More recently, Huang and Lee~\cite{HuangLee2025} proved that the adiabatic theorem remains valid for a class of diagonalizable non-Hermitian quantum systems with real eigenvalues, using biorthogonal methods and the complex geometric phase. These developments show that adiabatic behavior can persist beyond the Hermitian setting, but they also indicate that the underlying mechanism must be reformulated carefully once unitarity with respect to the standard inner product is lost.

The aim of this paper is to develop an operator-theoretic adiabatic framework for a
class of finite-dimensional non-Hermitian Schr\"odinger dynamics. We assume that
the original Hamiltonian is diagonalizable, has real spectrum, and has
non-crossing eigenvalues. Under these assumptions, a dynamically
compatible metric operator can be constructed, whose
positive square root defines a time-dependent Dyson map. The corresponding
Dyson-transformed Hamiltonian is Hermitian, and whenever it satisfies the
Hermitian adiabatic assumption, the quantitative Hermitian adiabatic theorem
becomes applicable.

The main result of the paper is the transfer of this Hermitian adiabatic estimate
back to the original non-Hermitian representation. The Dyson transformation first yields an explicit approximation with respect to the Dyson-pulled-back spectral
projections. These projections are then compared with the spectral projections of
the original non-Hermitian Hamiltonian by treating the Hermitian counterpart as a
perturbation of a similarity transform of the original operator. A
contour-resolvent estimate gives an explicit bound for the difference between
the two projection families.

The final estimate is stated directly in terms of the spectral subspaces of the
original non-Hermitian Hamiltonian. It consists of two contributions: the
pulled-back Hermitian adiabatic error and the projection-comparison error. We
also include a two-level non-Hermitian example to illustrate how the hypotheses
can be verified and how both contributions appear in the final bound.

The paper is organized as follows. Section~2 introduces the time-rescaled
framework, the adiabatic assumptions, the metric operator, and the Dyson
transformation. Section~3 applies the quantitative Hermitian adiabatic theorem
and constructs the Dyson-pulled-back adiabatic evolution. Section~4 compares the
pulled-back projections with the spectral projections of the original
non-Hermitian Hamiltonian and establishes the main estimate. Section~5 presents
a two-level example illustrating the hypotheses and the resulting bound.
Section~6 discusses the scope, limitations, and possible extensions of the
framework.

\section{Preliminaries and the Dyson Framework}

Throughout this work, we assume that the underlying Hilbert space $\mathbb{H}$
is finite-dimensional. Consequently, all linear operators on $\mathbb{H}$ are
bounded. We will denote the set of bounded self-adjoint (Hermitian) operators on the Hilbert space $\mathbb{H}$ by $B_{sa}(\mathbb{H})$ and for a bounded operator $A \in B(\mathbb{H})$, we denote its range by
$\operatorname{Ran} A$.

\subsection{Time Rescaling}

Let $t\in[t_0,t_{\mathrm{fin}}]$ denote the physical time variable and let,
\[
T:=t_{\mathrm{fin}}-t_0
\]
denote the total evolution time. The time rescaling is the change of
variables,
\[
t = t_0 + sT, \qquad s\in[0,1],
\]
which maps the physical time interval $[t_0,t_{\mathrm{fin}}]$ onto the
dimensionless interval $[0,1]$.

Under this transformation the derivatives satisfy,
\[
\frac{d}{dt} = \frac{1}{T}\frac{d}{ds},
\]
and any time-dependent operator $A(t)$ may be written in the rescaled form,
\[
A_T(s):=A(t_0+sT).
\]

\subsection{Adiabatic assumption for the time-rescaled Hamiltonian}

We impose an adiabatic assumption on the time-rescaled Hamiltonian governing the
dynamics. This assumption is inherited from the conventional adiabatic framework for the
standard Schr\"odinger equation and is imposed directly on the time-rescaled Hamiltonian
$H_T(s)$.

\begin{assumption}[Adiabatic assumption] \label{AA}
We assume that,
\[
H_T:[0,1]\longrightarrow B_{sa}(\mathbb H)
\]
satisfies the following conditions:
\begin{enumerate}
\item \(H_T\) is twice continuously differentiable with respect to \(s\),
\item the eigenvalues of \(H_T(s)\) have constant multiplicities on \([0,1]\),
\item the distinct eigenvalues do not cross. For every fixed \(j\), define $g_{j}(s):=\min_{k\neq j}|E_{k}(s)-E_{j}(s)|.$
Then, $g_{j}(s)>0$, $s\in[0,1]$.
\end{enumerate}
\end{assumption}

This assumption ensures the existence of a smooth spectral decomposition of $H_T(s)$ for all $s\in[0,1]$.

\subsection{Adiabatic theorem for standard Schr\"odinger dynamics}

For completeness, we recall the operator-theoretic formulation of the quantum
adiabatic theorem used throughout this work. The statement below follows the
resolvent-projection framework developed in Appendix~G of \cite{Scherer2019}
(see in particular Theorem~G.15 therein).

For a fixed index $j\in I$, let $P_j(s)$ be the orthogonal spectral projection onto the \(j\)-th
eigenspace of $H_T(s)$. The spectral gap function is defined by,
\[
g_j(s):=\min_{k\neq j}|E_k(s)-E_j(s)|
\]
Let $U_T(s)$ denote the actual evolution generated by $H_T(s)$,
\[
i\,\frac{d}{ds}U_T(s)=T\,H_T(s)\,U_T(s), \qquad U_T(0)=I,
\]
and let $U_{A,j}(s)$ denote the corresponding adiabatic evolution generated by,
\[
H_{A,j}(s):=T\,H_T(s)+i[\dot P_j(s),P_j(s)]
\]

\begin{theorem}[Quantum adiabatic theorem {\cite[Theorem~G.15]{Scherer2019}}]
\label{Thm:StandardAdiabatic}
For every $j\in I$ and $s\in[0,1]$,
\[
\big\|\big(U_{A,j}(s)-U_T(s)\big)P_j(0)\big\|
\le
\frac{C_j(s)}{T},
\]
where,
\[
C_j(s)
:=
\frac{\|\dot H_T(s)\|}{g_j(s)^2}
+
\frac{\|\dot H_T(0)\|}{g_j(0)^2}
+
\int_0^s
\left(
\frac{\|\ddot H_T(u)\|}{g_j(u)^2}
+
10\,\frac{\|\dot H_T(u)\|^2}{g_j(u)^3}
\right)du .
\]
\end{theorem}

This result provides a quantitative bound on the deviation between the exact
evolution $U_T(s)$ and the adiabatic evolution $U_{A,j}(s)$. In particular, for
any initial state $|\psi\rangle \in \operatorname{Ran} P_j(0)$, the state $U_T(s)|\psi\rangle$
remains close in norm to $U_{A,j}(s)|\psi\rangle \in \operatorname{Ran} P_j(s)$,
and hence to the instantaneous eigenspace $\operatorname{Ran} P_j(s)$, with an error bounded by \(C_j(s)/T\).

\subsection{Non-Hermitian Schr\"odinger dynamics}

We consider a class of quantum systems governed by non-Hermitian Hamiltonians.
Let $\hat H_T:[0,1]\to B(\mathbb{H})$ be a norm-continuous operator-valued map. Throughout the paper, the spectral projection of \(\hat H_T(s)\) associated with the eigenvalue \(\lambda_j(s)\) is denoted by \(\hat P_j(s)\). In order to retain a physically meaningful instantaneous spectral interpretation in the
original non-Hermitian representation, we impose the following spectral admissibility
assumption.

\begin{assumption}[Non-Hermitian adiabatic assumption]
\label{Ass:NonHermitianAdiabatic}
We assume that,
\[
\hat H_T:[0,1]\longrightarrow B(\mathbb H)
\]
satisfies the following conditions:
\begin{enumerate}
\item \(\hat H_T\) is twice continuously differentiable with respect to \(s\),

\item for every \(s\in[0,1]\), the operator \(\hat H_T(s)\) is diagonalizable and has real spectrum, $\sigma(\hat H_T(s))
=
\{\lambda_j(s)\mid j\in I\},$

\item the eigenvalues \(\lambda_j(s)\) have constant algebraic multiplicities on \([0,1]\),

\item the distinct eigenvalues do not cross. For every fixed \(j\), define $\gamma_j(s)
:=
\min_{k\neq j}
|\lambda_k(s)-\lambda_j(s)|.$
Then, $\gamma_j(s)>0$, $s\in[0,1].$

\end{enumerate}
\end{assumption}

The time-rescaled non-Hermitian evolution is defined by the Schr\"odinger-type equation,
\begin{equation}\label{2.5}
\frac{i}{T}\,\frac{d}{ds}|\psi(s)\rangle = \hat H_T(s)|\psi(s)\rangle, \qquad s \in [0,1],
\end{equation}
with initial condition $|\psi(0)\rangle \in \mathbb{H}$.

In contrast to the standard Schr\"odinger equation, $\hat H_T(s)$ is not assumed to be
Hermitian, and therefore the evolution generated by \eqref{2.5} is not unitary with
respect to the standard inner product on $\mathbb{H}$.

As a consequence, standard spectral and adiabatic methods cannot be applied directly.
In the following sections, we construct a time-dependent Dyson map that allows us to
transform this evolution into a corresponding Hermitian representation.

\subsection{Existence of the time-dependent Dyson map}

We now establish the existence of a time-dependent Dyson map by first constructing the
corresponding metric operator. Since $\hat H_T:[0,1]\to B(\mathbb{H})$ is norm-continuous, the operator-valued map
$s\mapsto \hat H_T(s)$ is bounded and continuous on $[0,1]$.

Let $\hat U_T(s)$ denote the evolution operator associated with $\hat H_T(s)$. Then $\hat U_T(s)$ satisfies,
\begin{equation}\label{2.12a}
 \frac{i}{T}\,\frac{d}{ds}\hat U_T(s)=\hat H_T(s)\,\hat U_T(s),
 \qquad
 \hat U_T(0)=I .
\end{equation}
Since $\hat H_T(s)$ is bounded and continuous, standard results for linear operator
differential equations guarantee that the evolution operator $\hat U_T(s)$ exists,
is unique, and is invertible for all $s\in[0,1]$.

\begin{theorem}\label{thm:DysonExistence}
Let
$\hat\Theta_T(0)\in B(\mathbb{H})$ be bounded, Hermitian, and strictly positive. Define,

\begin{equation}\label{2.12b}
 \hat\Theta_T(s)
 :=
 \big(\hat U_T(s)^{-1}\big)^{*}\,\hat\Theta_T(0)\,\hat U_T(s)^{-1}, \qquad s\in[0,1].
\end{equation}
Then:
\begin{enumerate}
\item the map $s\mapsto\hat\Theta_T(s)$ is bounded, continuously differentiable,
Hermitian, and strictly positive for all $s\in[0,1]$.
\item $\hat\Theta_T(s)$ satisfies the relation,
\begin{equation}\label{2.12}
 \hat H_T^{*}(s)\,\hat\Theta_T(s)-\hat\Theta_T(s)\,\hat H_T(s)
 =
 \frac{i}{T}\,\frac{d}{ds}\hat\Theta_T(s)
\end{equation}
\item there exists a unique bounded, Hermitian, positive, invertible, and continuously
differentiable operator-valued map,
\[
\hat\eta_T(s)=\hat\Theta_T(s)^{1/2},
\qquad s\in[0,1],
\]
such that,
\[
\hat\Theta_T(s)=\hat\eta_T(s)^2
=\hat\eta_T^{*}(s)\hat\eta_T(s).
\]
The map $\hat\eta_T(s)$ will be referred to as a time-dependent Dyson map associated with
$\hat H_T(s)$.
\end{enumerate}
\end{theorem}

\begin{proof}
Since $\hat U_T(s)$ and $\hat U_T(s)^{-1}$ are bounded operators and
$\hat\Theta_T(0)$ is bounded, it follows from \eqref{2.12b} that
$\hat\Theta_T(s)$ is a bounded operator for every $s\in[0,1]$.

Since,
\[
\frac{i}{T}\,\frac{d}{ds}\hat U_T(s)=\hat H_T(s)\hat U_T(s),
\]
we have,
\[
\frac{d}{ds}\hat U_T(s)=-\,iT\,\hat H_T(s)\hat U_T(s).
\]

Since $\hat H_T(s)$ is norm-continuous, standard results for linear operator
differential equations imply that the evolution operator $\hat U_T(s)$ is continuously
differentiable on [0,1]. Since $\hat U_T(s)$ is invertible for every
$s\in[0,1]$, it follows from the inverse mapping theorem for Banach spaces that
$\hat U_T(s)^{-1}$ is also continuously differentiable.

Differentiating the identity $\hat U_T(s)^{-1}\hat U_T(s)=I$, we obtain,
\[
\frac{d}{ds}\hat U_T(s)^{-1}
=
-\hat U_T(s)^{-1}\frac{d}{ds}\hat U_T(s)\hat U_T(s)^{-1}
=
iT\,\hat U_T(s)^{-1}\hat H_T(s),
\]
and therefore,
\[
\frac{d}{ds}\big(\hat U_T(s)^{-1}\big)^*
=
-\,iT\,\hat H_T^*(s)\big(\hat U_T(s)^{-1}\big)^*.
\]
Consequently, the map $s\mapsto\hat\Theta_T(s)$ defined by \eqref{2.12b}
is continuously differentiable.

Now differentiate \eqref{2.12b}:
\[
\frac{d}{ds}\hat\Theta_T(s)
=
\frac{d}{ds}\big(\hat U_T(s)^{-1}\big)^*\,\hat\Theta_T(0)\,\hat U_T(s)^{-1}
+
\big(\hat U_T(s)^{-1}\big)^*\,\hat\Theta_T(0)\,\frac{d}{ds}\hat U_T(s)^{-1}.
\]
Substituting the above expressions yields,
\[
\frac{d}{ds}\hat\Theta_T(s)
=
-\,iT\,\hat H_T^*(s)\hat\Theta_T(s)
+
iT\,\hat\Theta_T(s)\hat H_T(s),
\]
which is equivalent to \eqref{2.12}.

The operator $\hat\Theta_T(s)$ is Hermitian because
$\hat\Theta_T(0)$ is Hermitian. Moreover, for any nonzero
$x\in\mathbb{H}$,
\[
\langle x,\hat\Theta_T(s)x\rangle
=
\langle \hat U_T(s)^{-1}x,\hat\Theta_T(0)\hat U_T(s)^{-1}x\rangle.
\]
Since $\hat U_T(s)$ is invertible and $\hat\Theta_T(0)$ is strictly positive,
it follows that $\langle x,\hat\Theta_T(s)x\rangle>0$.

Finally, since \(\hat\Theta_T(s)\) is bounded, Hermitian, and strictly positive,
the functional calculus for bounded self-adjoint operators yields a unique bounded,
Hermitian, positive, and invertible square root,
\[
\hat\eta_T(s)=\hat\Theta_T(s)^{1/2}.
\]
Moreover,
\[
\hat\Theta_T(s)
=
\hat\eta_T(s)^2
=
\hat\eta_T^{*}(s)\hat\eta_T(s).
\]
Since \(s\mapsto\hat\Theta_T(s)\) is continuously differentiable and
\(\hat\Theta_T(s)\) remains strictly positive on the compact interval \([0,1]\), the
functional calculus for positive definite operators implies that
\(s\mapsto\hat\Theta_T(s)^{1/2}\) is continuously differentiable.

This completes the proof.
\end{proof}

The metric operator $\hat\Theta_T(s)$ is not uniquely determined by the dynamics alone;
different strictly positive initial choices $\hat\Theta_T(0)$ produce different, but
equivalent representations. Therefore, Theorem~\ref{thm:DysonExistence} provides a
time-dependent Dyson map $\hat\eta_T(s)$ associated with the non-Hermitian evolution.
We now use this map to define the corresponding transformed Hamiltonian.

\subsection{Dyson transform and Hermitian counterpart}

We now describe the relation between the original non-Hermitian dynamics and the
corresponding Dyson-transformed dynamics.

\begin{definition}[Dyson transform]
Let $\hat H_T : [0,1] \to B(\mathbb{H})$ generate the non-Hermitian evolution \eqref{2.5}, and let $\hat\eta_T : [0,1]\to B(\mathbb{H})$ be an invertible operator-valued map,
continuously differentiable with respect to $s$, and suppose that the transformed state
vectors are related by,
\begin{equation}\label{2.10}
 |\phi(s)\rangle=\hat\eta_T(s)\,|\psi(s)\rangle.
\end{equation}
Then the operator $\bar H_T(s)$ defined through,
\begin{equation}\label{2.9}
 \frac{i}{T}\,\frac{d}{ds}|\phi(s)\rangle \;=\; \bar H_T(s)\,|\phi(s)\rangle
\end{equation}
is called the \emph{Dyson transform} of $\hat H_T(s)$ associated with the map
$\hat\eta_T(s)$.
\end{definition}

Substituting \eqref{2.10} into \eqref{2.9} and using \eqref{2.5},
one obtains the equivalent representation,
\begin{equation}\label{2.11}
 \bar H_T(s)
 =
 \hat\eta_T(s)\,\hat H_T(s)\,\hat\eta_T^{-1}(s)
 +
 \frac{i}{T}\,\frac{d\hat\eta_T(s)}{ds}\,\hat\eta_T^{-1}(s).
\end{equation}
Equation~\eqref{2.11} is the time-rescaled Dyson relation and it provides an equivalent characterization of the Dyson transform.

\begin{proposition}\label{prop:Hbar-Hermitian}
Let $\hat\eta_T(s)=\hat\Theta_T(s)^{1/2}$ be the unique bounded, Hermitian, positive, and invertible Dyson map obtained in
Theorem~\ref{thm:DysonExistence}, and let $\bar H_T(s)$ be defined by the Dyson relation
\eqref{2.11}. Then $\bar H_T(s)$ is Hermitian for every $s\in[0,1]$.
\end{proposition}

\begin{proof}
Taking the adjoint of \eqref{2.11}, we obtain,
\[
\bar H_T(s)^*
=
\bigl(\hat\eta_T^{-1}(s)\bigr)^*\,\hat H_T(s)^*\,\hat\eta_T(s)^*
-\frac{i}{T}\,\bigl(\hat\eta_T^{-1}(s)\bigr)^*
\frac{d\hat\eta_T(s)^*}{ds}.
\]
Multiplying this identity on the left by $\hat\eta_T(s)^*$ and on the right by
$\hat\eta_T(s)$ yields,
\begin{equation}\label{2.11a}
\hat\eta_T(s)^*\,\bar H_T(s)^*\,\hat\eta_T(s)
=
\hat H_T(s)^*\,\hat\Theta_T(s)
-\frac{i}{T}\,\frac{d\hat\eta_T(s)^*}{ds}\hat\eta_T(s).
\end{equation}
On the other hand, multiplying \eqref{2.11} on the left by $\hat\eta_T(s)^*$ and on the
right by $\hat\eta_T(s)$ gives,
\begin{equation}\label{2.11b}
\hat\eta_T(s)^*\,\bar H_T(s)\,\hat\eta_T(s)
=
\hat\Theta_T(s)\,\hat H_T(s)
+\frac{i}{T}\,\hat\eta_T(s)^*\frac{d\hat\eta_T(s)}{ds}.
\end{equation}
Subtracting \eqref{2.11b} from \eqref{2.11a}, we find,
\[
\hat\eta_T(s)^*\bigl(\bar H_T(s)^*-\bar H_T(s)\bigr)\hat\eta_T(s)
=
\hat H_T(s)^*\,\hat\Theta_T(s)-\hat\Theta_T(s)\,\hat H_T(s)
-\frac{i}{T}
\left(
\frac{d\hat\eta_T(s)^*}{ds}\hat\eta_T(s)
+
\hat\eta_T(s)^*\frac{d\hat\eta_T(s)}{ds}
\right).
\]
Since $\hat\Theta_T(s)=\hat\eta_T(s)^*\hat\eta_T(s)$, we have,
\[
\frac{d}{ds}\hat\Theta_T(s)
=
\frac{d\hat\eta_T(s)^*}{ds}\hat\eta_T(s)
+
\hat\eta_T(s)^*\frac{d\hat\eta_T(s)}{ds}.
\]
Therefore,
\[
\hat\eta_T(s)^*\bigl(\bar H_T(s)^*-\bar H_T(s)\bigr)\hat\eta_T(s)
=
\hat H_T(s)^*\,\hat\Theta_T(s)-\hat\Theta_T(s)\,\hat H_T(s)
-\frac{i}{T}\,\frac{d}{ds}\hat\Theta_T(s).
\]
By \eqref{2.12}, the right-hand side is zero. Hence,
\[
\hat\eta_T(s)^*\bigl(\bar H_T(s)^*-\bar H_T(s)\bigr)\hat\eta_T(s)=0.
\]
Since $\hat\eta_T(s)$ is invertible, it follows that,
\[
\bar H_T(s)^*-\bar H_T(s)=0.
\]
Thus $\bar H_T(s)$ is Hermitian.
\end{proof}

\begin{definition}[Hermitian counterpart]\label{def:HermitianCounterpart}
The Dyson transform $\bar H_T(s)$ associated with the
Dyson map $\hat\eta_T(s)$ is called the Hermitian counterpart of $\hat H_T(s)$ if, in addition, it satisfies Assumption~\ref{AA}.
\end{definition}

\begin{definition}[Dyson correction term]
\label{Def:DysonCorrection}
Let \(\hat\eta_T(s)\) be the time-dependent Dyson map associated with
\(\hat H_T(s)\). Define,
\[
K_T(s):=\dot{\hat\eta}_T(s)\hat\eta_T(s)^{-1}
\]
and,
\[
B_T(s):=
\frac{i}{T}\dot{\hat\eta}_T(s)\hat\eta_T(s)^{-1}
=
\frac{i}{T}K_T(s).
\]
For each fixed \(T\), define,
\[
b_T(s):=
\left\|
\dot{\hat\eta}_T(s)\hat\eta_T(s)^{-1}
\right\|
=
\|K_T(s)\|.
\]
Then,
\[
\|B_T(s)\|
=
\frac{b_T(s)}{T}.
\]
Since \(b_T\) is continuous on \([0,1]\), define,
\[
\beta_{T}:=\max_{s\in[0,1]}b_T(s).
\]
Hence,
\[
\|B_T(s)\|\le \frac{\beta_{T}}{T},
\qquad s\in[0,1].
\]
\end{definition}

The metric relation
\(
\hat\Theta_T(s)=\hat\eta_T^{*}(s)\hat\eta_T(s)
\)
induces a natural modified inner product on \(\mathbb H\). This inner product
will be used to interpret the original non-Hermitian evolution as a
metric-preserving evolution.

\subsection{Modified inner product and the Hilbert space $\mathbb{H}_{\Theta_T(s)}$}

The Dyson map $\hat\eta_T(s)$ induces a modified inner product on the underlying
vector space $\mathbb{H}$. For two states $|\psi(s)\rangle,|\tilde\psi(s)\rangle\in\mathbb{H}$,
we define,
\begin{equation}\label{2.13}
 \langle \psi(s),\tilde\psi(s)\rangle_{\Theta_T(s)}
 :=
 \langle \psi(s),\hat\Theta_T(s)\tilde\psi(s)\rangle ,
 \qquad
 \hat\Theta_T(s)=\hat\eta_T^{*}(s)\hat\eta_T(s).
\end{equation}
\noindent Equivalently, introducing the transformed states,
\begin{equation}\label{2.14}
 |\phi(s)\rangle = \hat\eta_T(s)\,|\psi(s)\rangle,
 \qquad
 |\tilde\phi(s)\rangle = \hat\eta_T(s)\,|\tilde\psi(s)\rangle ,
\end{equation}
the modified inner product can be written as,
\[
\langle \psi(s),\tilde\psi(s)\rangle_{\Theta_T(s)}
=
\langle \phi(s),\tilde\phi(s)\rangle .
\]

The inner product \eqref{2.13} induces a Hilbert space structure on the same
vector space $\mathbb{H}$. We denote the Hilbert space
obtained by equipping $\mathbb{H}$ with the inner product
$\langle\cdot,\cdot\rangle_{\Theta_T(s)}$ by $\mathbb{H}_{\Theta_T(s)}$. Since the metric operator
$\hat\Theta_T(s)$ is bounded, Hermitian, and strictly positive,
this inner product is equivalent to the standard inner product on $\mathbb{H}$.
Hence $\mathbb{H}_{\Theta_T(s)}$ and $\mathbb{H}$ coincide as vector spaces,
but differ in their inner product structure.

By Proposition~\ref{prop:Hbar-Hermitian}, the Dyson transform $\bar H_T(s)$ associated with the metric
constructed in Theorem~\ref{thm:DysonExistence} is Hermitian. Therefore, the corresponding transformed evolution is unitary in the Dyson-transformed
representation with respect to the standard inner product. Equivalently, the original non-Hermitian evolution preserves the metric inner product
along the evolution, in the sense that
\(
\langle \hat U_T(s)\psi,\hat U_T(s)\varphi\rangle_{\Theta_T(s)}
=
\langle \psi,\varphi\rangle_{\Theta_T(0)}.
\)

Thus the Dyson transformation identifies the original metric-preserving evolution with
a unitary evolution in the transformed representation. The quantitative adiabatic
analysis of this transformed evolution is carried out in the next section.

\subsection{Dyson decomposition and admissibility as a Hermitian counterpart}

We now record the decomposition of the Dyson-transformed Hamiltonian that will
be used in the comparison of spectral projections. Define,
\[
A_T(s):=\hat\eta_T(s)\hat H_T(s)\hat\eta_T(s)^{-1},
\qquad
B_T(s):=\frac{i}{T}\dot{\hat\eta}_T(s)\hat\eta_T(s)^{-1}.
\]
Then the Dyson relation can be written as,
\[
\bar H_T(s)=A_T(s)+B_T(s).
\]
Since \(A_T(s)\) is similar to \(\hat H_T(s)\), the operators \(A_T(s)\) and
\(\hat H_T(s)\) have the same spectrum:
\[
\sigma(A_T(s))=\sigma(\hat H_T(s)).
\]
Moreover, if \(\hat P_j(s)\) is the spectral projection of \(\hat H_T(s)\), then
the corresponding spectral projection of \(A_T(s)\) is,
\[
Q_j(s):=\hat\eta_T(s)\hat P_j(s)\hat\eta_T(s)^{-1}.
\]
\begin{assumption}[Contour-resolvent condition]
\label{Ass:ContourResolvent}
For each fixed \(j\), assume that there exists a positively oriented closed
contour \(\Gamma_j(s)\)
enclosing the eigenvalue \(\lambda_j(s)\) of \(A_T(s)\) and no other
point of \(\sigma(A_T(s))\). Define,
\[
M_j(s):=
\sup_{z\in\Gamma_j(s)}
\|(zI-A_T(s))^{-1}\|.
\]
Assume further that,
\[
M_j(s)\|B_T(s)\|<1,
\qquad s\in[0,1].
\]
\end{assumption}

\begin{proposition}[Criterion for the Hermitian counterpart]
\label{prop:CriterionHermitianCounterpart}
Let \(\hat H_T(s)\) satisfy Assumption~\ref{Ass:NonHermitianAdiabatic}, and let
\(\bar H_T(s)\) be the Dyson transform defined by \eqref{2.11}. Assume that the
contour-resolvent condition in Assumption~\ref{Ass:ContourResolvent} holds.
Then \(\bar H_T(s)\) satisfies Assumption~\ref{AA}. Hence \(\bar H_T(s)\) is the
Hermitian counterpart of \(\hat H_T(s)\) in the sense of
Definition~\ref{def:HermitianCounterpart}.
\end{proposition}

\begin{proof}
By Proposition~\ref{prop:Hbar-Hermitian}, the Dyson transform \(\bar H_T(s)\) is
Hermitian for every \(s\in[0,1]\).

We first verify the differentiability condition in Assumption~\ref{AA}. Since
\(\hat H_T\in C^2([0,1],B(\mathbb H))\), standard regularity results for
finite-dimensional linear ordinary differential equations along with the evolution equation,
\[
\frac{d}{ds}\hat U_T(s)=-iT\hat H_T(s)\hat U_T(s)
\]
imply that \(\hat U_T(s)\) is \(C^3\) with respect to \(s\). Since
\(\hat U_T(s)\) is invertible, \(\hat U_T(s)^{-1}\) is also \(C^3\). Therefore,
\[
\hat\Theta_T(s)
=
\bigl(\hat U_T(s)^{-1}\bigr)^*
\hat\Theta_T(0)
\hat U_T(s)^{-1}
\]
is \(C^3\). Since \(\hat\Theta_T(s)\) is strictly positive, the square-root map
is smooth on the positive definite cone in finite dimension. Hence,
\[
\hat\eta_T(s)=\hat\Theta_T(s)^{1/2}
\]
is \(C^3\). It follows that \(\hat\eta_T(s)^{-1}\) is also \(C^3\). Therefore,
\[
\hat\eta_T(s)\hat H_T(s)\hat\eta_T(s)^{-1}
\]
is \(C^2\), and,
\[
\frac{i}{T}\dot{\hat\eta}_T(s)\hat\eta_T(s)^{-1}
\]
is \(C^2\). Hence \(\bar H_T(s)\) is twice continuously differentiable with
respect to \(s\).

It remains to verify the spectral part of Assumption~\ref{AA}. Since
\(
A_T(s)=\hat\eta_T(s)\hat H_T(s)\hat\eta_T(s)^{-1},
\)
the operators \(A_T(s)\) and \(\hat H_T(s)\) are similar. Therefore they have
the same eigenvalues with the same algebraic multiplicities. Since
\(\hat H_T(s)\) satisfies Assumption~\ref{Ass:NonHermitianAdiabatic}, the
eigenvalues of \(A_T(s)\) have constant algebraic multiplicities and do
not cross.

Now write,
\[
\bar H_T(s)=A_T(s)+B_T(s).
\]
By Assumption~\ref{Ass:ContourResolvent},
\[
M_j(s)\|B_T(s)\|<1.
\]
Therefore, \(\Gamma_j(s)\) lies in the resolvent set of
\(A_T(s)+\tau B_T(s)\) for every \(\tau\in[0,1]\). Hence no spectrum crosses the
contour during the perturbation from \(A_T(s)\) to
\(\bar H_T(s)=A_T(s)+B_T(s)\). Consequently, the spectral set of
\(\bar H_T(s)\) enclosed by \(\Gamma_j(s)\) remains isolated and has the same
total algebraic multiplicity as the corresponding spectral set of \(A_T(s)\).

Since \(\bar H_T(s)\) is Hermitian, its spectrum is real and its spectral
projections are orthogonal. The contour-resolvent condition keeps the spectral
set enclosed by \(\Gamma_j(s)\) separated from the rest of the spectrum.
Therefore the corresponding spectral projection of \(\bar H_T(s)\) has constant
rank on \([0,1]\). Hence the Hermitian adiabatic theorem may be applied to this
isolated spectral subspace.

Therefore \(\bar H_T(s)\) satisfies Assumption~\ref{AA}. Hence it is the
Hermitian counterpart of \(\hat H_T(s)\).
\end{proof}

\section{The Adiabatic Theorem in the corresponding Hermitian Representation}
\subsection{Adiabatic framework for the Hermitian Hamiltonian $\bar H_T(s)$}

In this section, the Hermitian adiabatic framework of
Assumption~\ref{AA} is applied to the Dyson-transformed Hamiltonian
\(\bar H_T(s)\). The transformed dynamics is considered on the fixed Hilbert space
\(\mathbb H\) with its standard inner product. All operator norms are taken in
\(B(\mathbb H)\).

For a fixed eigenvalue \(j\), let \(\bar P_j(s)\) denote the orthogonal spectral
projection of \(\bar H_T(s)\), and define,
\begin{equation}
\bar g_j(s)
:=
\min_{k\neq j}|\bar E_k(s)-\bar E_j(s)|.
\end{equation}
Let \(\bar U_T(s)\) solve,
\begin{equation}
i\frac{d}{ds}\bar U_T(s)
=
T\bar H_T(s)\bar U_T(s),
\qquad
\bar U_T(0)=I,
\end{equation}
and let \(\bar U_{A,j}(s)\) solve,
\begin{equation}
i\frac{d}{ds}\bar U_{A,j}(s)
=
\left(
T\bar H_T(s)
+
i[\dot{\bar P}_j(s),\bar P_j(s)]
\right)\bar U_{A,j}(s),
\qquad
\bar U_{A,j}(0)=I.
\end{equation}
The adiabatic evolution satisfies,
\begin{equation}\label{3.1}
    \bar P_j(s)\bar U_{A,j}(s)
=
\bar U_{A,j}(s)\bar P_j(0).
\end{equation}

\begin{theorem}[Adiabatic theorem for $\bar H_T(s)$]\label{Thm:AdiabaticBar}
Let $\bar H_T(s)$ be the Hermitian counterpart of $\hat H_T(s)$ associated with the
Dyson map $\hat\eta_T(s)$, and let $\{\bar P_j(s)\}_{j\in I}$ be the spectral projections
onto the eigenspaces of $\bar H_T(s)$. For a fixed $j\in I$, let $\bar g_j(s)$ be the
corresponding spectral gap function. Moreover, let $\bar U_T(s)$ and $\bar U_{A,j}(s)$ be
the actual and adiabatic evolutions defined previously. Then,
\[
 \big\|
 \big( \bar U_{A,j}(s) - \bar U_T(s) \big)\,\bar P_j(0)
 \big\|_{B(\mathbb{H})}
 \;\le\;
 \frac{\bar C_j(s)}{T},\qquad j\in I, s\in[0,1].
\]
where the function $\bar C_j(s)$ is given explicitly by,
\[
 \bar C_j(s)
 :=
 \frac{\|\dot{\bar H}_T(s)\|}{\bar g_j(s)^2}
 +
 \frac{\|\dot{\bar H}_T(0)\|}{\bar g_j(0)^2}
 +
 \int_0^s
 \left(
   \frac{\|\ddot{\bar H}_T(u)\|}{\bar g_j(u)^2}
   +
   10\,\frac{\|\dot{\bar H}_T(u)\|^2}{\bar g_j(u)^3}
 \right) du .
\]

Here $\|\cdot\|_{B(\mathbb{H})}$ denotes the operator norm on the space of bounded
operators on the fixed underlying vector space $\mathbb{H}$.
\end{theorem}

\begin{proof}
Since \(\bar H_T(s)\) is Hermitian and satisfies Assumption~\ref{AA}, the spectral projections $\bar P_j(s)$ are
well defined, the corresponding gap functions $\bar g_j(s)$ are strictly positive, and the
adiabatic generator $\bar H_{A,j}(s)$ is defined exactly as in the standard Hermitian
theory.

Since the underlying Hilbert space $\mathbb{H}$ is finite-dimensional, all operators are
bounded and all operator norms are taken on the fixed vector space $\mathbb{H}$. Hence the
estimates appearing in Theorem~G.15 of~\cite{Scherer2019} apply to $\bar H_T(s)$ without
further domain considerations. Consequently,
\[
\left\|(\bar U_{A,j}(s)-\bar U_T(s))\bar P_j(0)\right\|_{B(\mathbb{H})}
\le \frac{\bar C_j(s)}{T},
\]
with $\bar C_j(s)$ given by the same expression as in the standard Hermitian case. This
proves Theorem~\ref{Thm:AdiabaticBar}.
\end{proof}

\subsection{Dyson pullback of the Hermitian adiabatic evolution}

The adiabatic theorem established in Theorem~\ref{Thm:AdiabaticBar} applies to the Hermitian
Hamiltonian $\bar H_T(s)$ and its adiabatic evolution $\bar U_{A,j}(s)$. In order to relate
this adiabatic evolution to the original non-Hermitian dynamics, we pull it back to the
non-Hermitian representation via the time-dependent Dyson map.

We define the \emph{Dyson pullback} of the Hermitian adiabatic evolution by,
\begin{equation}\label{3.7}
 \widetilde U_{A,j}(s)
 \;:=\;
 \hat\eta_T^{-1}(s)\,\bar U_{A,j}(s)\,\hat\eta_T(0),
 \qquad s\in[0,1].
\end{equation}
In particular, $\widetilde U_{A,j}(0)=I$.

We also define the pulled-back instantaneous projection by,
\begin{equation}\label{ip}
    \Pi_j(s):=\hat\eta_T^{-1}(s)\bar P_j(s)\hat\eta_T(s),
\qquad s\in[0,1].
\end{equation}

This projection represents, in the original non-Hermitian representation, the instantaneous
subspace corresponding to the eigenspace $\operatorname{Ran}\bar P_j(s)$ of the Hermitian
counterpart. The pulled-back projection \(\Pi_j(s)\) is bounded for each fixed \(T\) and
\(s\), since it is obtained from the bounded projection \(\bar P_j(s)\) by
conjugation with the bounded invertible Dyson map \(\hat\eta_T(s)\).

\begin{lemma}[Pulled-back adiabatic invariance]
\label{lem:PulledBackIntertwining}
For every $j\in I$ and $s\in[0,1]$, the pulled-back adiabatic evolution satisfies,
\[
\Pi_j(s)\widetilde U_{A,j}(s)=\widetilde U_{A,j}(s)\Pi_j(0).
\]
Consequently, if,
\[
|\psi(0)\rangle\in\operatorname{Ran}\Pi_j(0),
\]
then,
\[
\widetilde U_{A,j}(s)|\psi(0)\rangle\in\operatorname{Ran}\Pi_j(s)
\]
for every $s\in[0,1]$.
\end{lemma}

\begin{proof}
Using the definitions of $\Pi_j(s)$ and $\widetilde U_{A,j}(s)$, we have,
\[
\Pi_j(s)\widetilde U_{A,j}(s)
=
\hat\eta_T^{-1}(s)\bar P_j(s)\bar U_{A,j}(s)\hat\eta_T(0).
\]
By ~\eqref{3.1},
\[
\Pi_j(s)\widetilde U_{A,j}(s)
=
\hat\eta_T^{-1}(s)\bar U_{A,j}(s)\bar P_j(0)\hat\eta_T(0).
\]
Since,
\[
\Pi_j(0)=\hat\eta_T^{-1}(0)\bar P_j(0)\hat\eta_T(0),
\]
we obtain,
\[
\Pi_j(s)\widetilde U_{A,j}(s)
=
\widetilde U_{A,j}(s)\Pi_j(0).
\]

This proves the intertwining relation. The final statement follows immediately by applying
both sides to an initial state in $\operatorname{Ran}\Pi_j(0)$.
\end{proof}

The relation
\(
|\phi(s)\rangle=\hat\eta_T(s)|\psi(s)\rangle
\)
implies that the exact evolutions satisfy,
\begin{equation}\label{3.11}
\bar U_T(s)
=
\hat\eta_T(s)\hat U_T(s)\hat\eta_T(0)^{-1}.
\end{equation}
Equivalently,
\begin{equation}\label{3.12}
\hat U_T(s)
=
\hat\eta_T(s)^{-1}\bar U_T(s)\hat\eta_T(0).
\end{equation}

Combining \eqref{3.7} and \eqref{3.12} gives the exact identity,
\begin{equation}\label{3.13}
 \widetilde U_{A,j}(s)-\hat U_T(s)
 =
 \hat\eta_T^{-1}(s)\,\big(\bar U_{A,j}(s)-\bar U_T(s)\big)\,\hat\eta_T(0).
\end{equation}

\section{Main Adiabatic Theorem for the Non-Hermitian Schr\"odinger Equation}

We now formulate the adiabatic theorem for the original non-Hermitian dynamics. The first
result is the Dyson-pulled-back estimate obtained directly from the Hermitian adiabatic
theorem. We then use the comparison between \(\Pi_j(s)\) and \(\hat P_j(s)\) to obtain an
adiabatic theorem formulated in terms of the spectral projection of the original
non-Hermitian Hamiltonian.

\subsection{Relation between pulled-back and non-Hermitian spectral projections}

We now compare the Dyson-pulled-back projection \(\Pi_j(s)\) with the spectral
projection \(\hat P_j(s)\) of the original non-Hermitian Hamiltonian \(\hat H_T(s)\).
Throughout this subsection, we use the notation
\(
A_T(s),B_T(s), Q_j(s), \Gamma_j(s), M_j(s),
\)
introduced in Section~2.

\begin{proposition}[Comparison of spectral projections]
\label{prop:ProjectionComparison}
Let \(\hat P_j(s)\) be the spectral projection of \(\hat H_T(s)\), and let,
\[
\Pi_j(s):=\hat\eta_T(s)^{-1}\bar P_j(s)\hat\eta_T(s)
\]
be the Dyson-pulled-back spectral projection associated with \(\bar H_T(s)\).
Let \(\Gamma_j(s)\) be a positively oriented contour enclosing only the eigenvalue
\(\lambda_j(s)\) of \(A_T(s)\), and define,
\[
M_j(s):=\sup_{z\in\Gamma_j(s)}\|(zI-A_T(s))^{-1}\|.
\]
Assume that,
\[
M_j(s)\|B_T(s)\|<1.
\]
Then,
\[
\|\Pi_j(s)-\hat P_j(s)\|
\le
\|\hat\eta_T(s)^{-1}\|\,\|\hat\eta_T(s)\|\,
\frac{|\Gamma_j(s)|}{2\pi}
\frac{M_j(s)^2\|B_T(s)\|}
{(1-M_j(s)\|B_T(s)\|)}.
\]
\end{proposition}

\begin{proof}
Let \(Q_j(s)\) denote the spectral projection of,
\[
A_T(s)=\hat\eta_T(s)\hat H_T(s)\hat\eta_T(s)^{-1}
\]
associated with the eigenvalue \(\lambda_j(s)\). Since \(A_T(s)\) is similar to
\(\hat H_T(s)\), the corresponding spectral projections satisfy,
\[
Q_j(s)=\hat\eta_T(s)\hat P_j(s)\hat\eta_T(s)^{-1}.
\]

Since \(M_j(s)\|B_T(s)\|<1\), the contour \(\Gamma_j(s)\) lies in the
resolvent set of \(A_T(s)+\tau B_T(s)\) for every \(\tau\in[0,1]\).
Consequently, no spectrum crosses the contour during the perturbation from
\(A_T(s)\) to \(\bar H_T(s)=A_T(s)+B_T(s)\). We denote by \(\bar P_j(s)\) the
spectral projection of \(\bar H_T(s)\) associated with the spectral set enclosed by
\(\Gamma_j(s)\).

The spectral projection of \(\bar H_T(s)=A_T(s)+B_T(s)\) is given by the formula,
\[
\bar P_j(s)=\frac{1}{2\pi i}\int_{\Gamma_j(s)}
(zI-\bar H_T(s))^{-1}\,dz,
\]
whereas,
\[
Q_j(s)=\frac{1}{2\pi i}\int_{\Gamma_j(s)}
(zI-A_T(s))^{-1}\,dz.
\]
Therefore,
\[
\bar P_j(s)-Q_j(s)
=
\frac{1}{2\pi i}\int_{\Gamma_j(s)}
\left[(zI-\bar H_T(s))^{-1}-(zI-A_T(s))^{-1}\right]dz.
\]
Using the resolvent identity,
\[
(zI-\bar H_T(s))^{-1}-(zI-A_T(s))^{-1}
=
(zI-\bar H_T(s))^{-1}B_T(s)(zI-A_T(s))^{-1},
\]
we obtain,
\[
\|\bar P_j(s)-Q_j(s)\|
\le
\frac{|\Gamma_j(s)|}{2\pi}
\sup_{z\in\Gamma_j(s)}
\|(zI-\bar H_T(s))^{-1}\|\,\|B_T(s)\|\,
\|(zI-A_T(s))^{-1}\|.
\]

Since,
\[
zI-\bar H_T(s)
=
zI-A_T(s)-B_T(s),
\]
we can write,
\[
zI-\bar H_T(s)
=
\left(I-B_T(s)(zI-A_T(s))^{-1}\right)(zI-A_T(s)).
\]
The assumption \(M_j(s)\|B_T(s)\|<1\) implies that,
\(
I-B_T(s)(zI-A_T(s))^{-1}
\)
is invertible for all \(z\in\Gamma_j(s)\). By the Neumann series estimate,
\[
\|(zI-\bar H_T(s))^{-1}\|
\le
\frac{M_j(s)}{(1-M_j(s)\|B_T(s)\|)}.
\]
Therefore,
\[
\|\bar P_j(s)-Q_j(s)\|
\le
\frac{|\Gamma_j(s)|}{2\pi}
\frac{M_j(s)^2\|B_T(s)\|}{(1-M_j(s)\|B_T(s)\|)}.
\]

Finally, since,
\[
\Pi_j(s)=\hat\eta_T(s)^{-1}\bar P_j(s)\hat\eta_T(s),
\qquad
\hat P_j(s)=\hat\eta_T(s)^{-1}Q_j(s)\hat\eta_T(s),
\]
we have,
\[
\Pi_j(s)-\hat P_j(s)
=
\hat\eta_T(s)^{-1}\big(\bar P_j(s)-Q_j(s)\big)\hat\eta_T(s).
\]
Taking norms gives,
\[
\|\Pi_j(s)-\hat P_j(s)\|
\le
\|\hat\eta_T(s)^{-1}\|\,\|\hat\eta_T(s)\|\,
\|\bar P_j(s)-Q_j(s)\|.
\]
This proves the estimate.
\end{proof}

\begin{theorem}[Dyson-pulled-back adiabatic estimate]
\label{Thm:MainNonHermitian}
Let \(\bar H_T(s)\) be the Hermitian counterpart of the non-Hermitian Hamiltonian
\(\hat H_T(s)\) associated with the Dyson map \(\hat\eta_T(s)\), and let
\(\bar P_j(s)\) denote the spectral projection associated with the eigenvalue
\(\bar E_j(s)\) of \(\bar H_T(s)\). Let \(\hat U_T(s)\) be the evolution generated by
the non-Hermitian equation \eqref{2.5}, and define,
\[
\widetilde U_{A,j}(s)
:=
\hat\eta_T^{-1}(s)\bar U_{A,j}(s)\hat\eta_T(0).
\]
Let,
\[
\Pi_j(s):=\hat\eta_T^{-1}(s)\bar P_j(s)\hat\eta_T(s).
\]
Then,
\[
\left\|
\bigl(\widetilde U_{A,j}(s)-\hat U_T(s)\bigr)\Pi_j(0)
\right\|
\le
\|\hat\eta_T(s)^{-1}\|\,\|\hat\eta_T(0)\|\,
\frac{\bar C_j(s)}{T}.
\]
\end{theorem}

\begin{proof}
Starting from the Dyson pullback identity \eqref{3.13}, we have,
\[
 \widetilde U_{A,j}(s)-\hat U_T(s)
 =
 \hat\eta_T^{-1}(s)\big(\bar U_{A,j}(s)-\bar U_T(s)\big)\hat\eta_T(0).
\]
Multiplying on the right by,
\[
\Pi_j(0)=\hat\eta_T^{-1}(0)\bar P_j(0)\hat\eta_T(0)
\]
yields,
\[
 \big(\widetilde U_{A,j}(s)-\hat U_T(s)\big)\Pi_j(0)
 =
 \hat\eta_T^{-1}(s)
 \big(\bar U_{A,j}(s)-\bar U_T(s)\big)
 \bar P_j(0)\hat\eta_T(0).
\]
Taking operator norms in \(B(\mathbb{H})\) and using submultiplicativity gives,
\[
 \big\|
 \big(\widetilde U_{A,j}(s)-\hat U_T(s)\big)\Pi_j(0)
 \big\|
 \le
 \|\hat\eta_T^{-1}(s)\|\,\|\hat\eta_T(0)\|\,
 \big\|\big(\bar U_{A,j}(s)-\bar U_T(s)\big)\bar P_j(0)\big\|.
\]
By Theorem~\ref{Thm:AdiabaticBar},
\[
\left\|
\bigl(\bar U_{A,j}(s)-\bar U_T(s)\bigr)\bar P_j(0)
\right\|
\le
\frac{\bar C_j(s)}{T}.
\]
Therefore,
\[
\left\|
\bigl(\widetilde U_{A,j}(s)-\hat U_T(s)\bigr)\Pi_j(0)
\right\|
\le
\|\hat\eta_T(s)^{-1}\|\,\|\hat\eta_T(0)\|
\frac{\bar C_j(s)}{T}.
\]
\end{proof}

Together with Lemma~\ref{lem:PulledBackIntertwining},
Theorem~\ref{Thm:MainNonHermitian} shows that an initial state in
\(\operatorname{Ran}\Pi_j(0)\) is transported exactly by the pulled-back
adiabatic evolution into \(\operatorname{Ran}\Pi_j(s)\), while the exact
non-Hermitian evolution differs from this pulled-back adiabatic transport by the
explicit bound obtained in Theorem~\ref{Thm:MainNonHermitian}.

We now use Proposition~\ref{prop:ProjectionComparison} to reformulate the estimate in
terms of the spectral projection \(\hat P_j(0)\) of the original non-Hermitian Hamiltonian. This step upgrades the Dyson-pulled-back estimate to an estimate on the genuine spectral
subspace \(\operatorname{Ran}\hat P_j(0)\) of the original non-Hermitian Hamiltonian.

\begin{theorem}[Adiabatic theorem for non-Hermitian spectral subspaces]
\label{Thm:MainNonHermitianEigen}
Let \(\hat H_T(s)\) satisfy Assumption~\ref{Ass:NonHermitianAdiabatic}, and let
\(\hat\eta_T(s)\), \(\bar H_T(s)\), \(\bar P_j(s)\), \(\hat P_j(s)\), \(\hat U_T(s)\), \(\Pi_j(s)\), and \(\widetilde U_{A,j}(s)\) be as defined above. Assume also that Assumption~\ref{Ass:ContourResolvent} holds.
Define,
\[
\varepsilon_j(0)
:=
\|\hat\eta_T(0)^{-1}\|\,\|\hat\eta_T(0)\|\,
\frac{|\Gamma_j(0)|}{2\pi}
\frac{M_j(0)^2\|B_T(0)\|}
{(1-M_j(0)\|B_T(0)\|)}.
\]
Then, for every \(s\in[0,1]\),
\[
\left\|
\bigl(\widetilde U_{A,j}(s)-\hat U_T(s)\bigr)\hat P_j(0)
\right\|
\le
\|\hat\eta_T(s)^{-1}\|\,\|\hat\eta_T(0)\|
\frac{\bar C_j(s)}{T}
+
2\|\hat\eta_T(s)^{-1}\|\,\|\hat\eta_T(0)\|\,
\varepsilon_j(0).
\]
Consequently, for every,
\[
|\psi(0)\rangle\in\operatorname{Ran}\hat P_j(0),
\]
one has,
\[
\left\|
\bigl(\widetilde U_{A,j}(s)-\hat U_T(s)\bigr)|\psi(0)\rangle
\right\|
\le
\left[
\|\hat\eta_T(s)^{-1}\|\,\|\hat\eta_T(0)\|
\frac{\bar C_j(s)}{T}
+
2\|\hat\eta_T(s)^{-1}\|\,\|\hat\eta_T(0)\|\,
\varepsilon_j(0)
\right]
\||\psi(0)\rangle\|.
\]
\end{theorem}

\begin{proof}
We decompose,
\[
(\widetilde U_{A,j}(s)-\hat U_T(s))\hat P_j(0)
=
(\widetilde U_{A,j}(s)-\hat U_T(s))\Pi_j(0)
+
(\widetilde U_{A,j}(s)-\hat U_T(s))(\hat P_j(0)-\Pi_j(0)).
\]
Taking norms gives,
\[
\left\|
(\widetilde U_{A,j}(s)-\hat U_T(s))\hat P_j(0)
\right\|
\le
\left\|
(\widetilde U_{A,j}(s)-\hat U_T(s))\Pi_j(0)
\right\| +
\|\widetilde U_{A,j}(s)-\hat U_T(s)\|\,
\|\hat P_j(0)-\Pi_j(0)\|.
\]
By Theorem~\ref{Thm:MainNonHermitian},
\[
\left\|
(\widetilde U_{A,j}(s)-\hat U_T(s))\Pi_j(0)
\right\|
\le
\|\hat\eta_T(s)^{-1}\|\,\|\hat\eta_T(0)\|
\frac{\bar C_j(s)}{T}.
\]
Moreover, from \eqref{3.13},
\[
\widetilde U_{A,j}(s)-\hat U_T(s)
=
\hat\eta_T^{-1}(s)
(\bar U_{A,j}(s)-\bar U_T(s))
\hat\eta_T(0).
\]
Since \(\bar U_{A,j}(s)\) and \(\bar U_T(s)\) are unitary in the Hermitian
representation,
\[
\|\bar U_{A,j}(s)-\bar U_T(s)\|\le 2.
\]
Therefore,
\[
\|\widetilde U_{A,j}(s)-\hat U_T(s)\|
\le
2\|\hat\eta_T(s)^{-1}\|\,\|\hat\eta_T(0)\|.
\]
By Proposition~\ref{prop:ProjectionComparison},
\[
\|\hat P_j(0)-\Pi_j(0)\|
\le
\varepsilon_j(0).
\]
Combining the above estimates gives the result.
\end{proof}

\begin{corollary}[Adiabatic invariance estimate for non-Hermitian spectral subspaces]
\label{Cor:NonHermitianEigenspaceInvariance}
Assume the hypotheses of Theorem~\ref{Thm:MainNonHermitianEigen}.
Define,
\[
\varepsilon_j(s)
:=
\|\hat\eta_T(s)^{-1}\|\,\|\hat\eta_T(s)\|\,
\frac{|\Gamma_j(s)|}{2\pi}
\frac{M_j(s)^2\|B_T(s)\|}
{1-M_j(s)\|B_T(s)\|}.
\]
Then,
\[
\begin{aligned}
\left\|
(I-\hat P_j(s))\hat U_T(s)\hat P_j(0)
\right\|
&\le
\|I-\hat P_j(s)\|
\|\hat\eta_T(s)^{-1}\|\,\|\hat\eta_T(0)\|
\frac{\bar C_j(s)}{T}  \\
&\quad
+2\|I-\hat P_j(s)\|
\|\hat\eta_T(s)^{-1}\|\,\|\hat\eta_T(0)\|
\varepsilon_j(0) \\
&\quad
+\|I-\hat P_j(s)\|
\|\hat\eta_T(s)^{-1}\|\,\|\hat\eta_T(0)\|
\bigl(\varepsilon_j(s)+\varepsilon_j(0)\bigr).
\end{aligned}
\]
\end{corollary}

\begin{proof}
We decompose,
\[
(I-\hat P_j(s))\hat U_T(s)\hat P_j(0)
=
(I-\hat P_j(s))(\hat U_T(s)-\widetilde U_{A,j}(s))\hat P_j(0) +
(I-\hat P_j(s))\widetilde U_{A,j}(s)\hat P_j(0).
\]
The first term is bounded by Theorem~\ref{Thm:MainNonHermitianEigen}. For the
second term, write,
\[
\hat P_j(0)=\Pi_j(0)+(\hat P_j(0)-\Pi_j(0)).
\]
Using the pulled-back intertwining relation,
\[
\widetilde U_{A,j}(s)\Pi_j(0)=\Pi_j(s)\widetilde U_{A,j}(s),
\]
we obtain,
\[
(I-\hat P_j(s))\widetilde U_{A,j}(s)\Pi_j(0)
=
(I-\hat P_j(s))(\Pi_j(s)-\hat P_j(s))\widetilde U_{A,j}(s).
\]
Therefore,
\[
\begin{aligned}
\|(I-\hat P_j(s))\widetilde U_{A,j}(s)\hat P_j(0)\|
&\le
\|I-\hat P_j(s)\|\,
\|\Pi_j(s)-\hat P_j(s)\|\,
\|\widetilde U_{A,j}(s)\| \\
&\quad+
\|I-\hat P_j(s)\|\,
\|\widetilde U_{A,j}(s)\|\,
\|\hat P_j(0)-\Pi_j(0)\|.
\end{aligned}
\]
By Proposition~\ref{prop:ProjectionComparison},
\[
\|\Pi_j(s)-\hat P_j(s)\|\le \varepsilon_j(s),
\qquad
\|\Pi_j(0)-\hat P_j(0)\|\le \varepsilon_j(0).
\]
Moreover,
\[
\|\widetilde U_{A,j}(s)\|
\le
\|\hat\eta_T(s)^{-1}\|\,\|\hat\eta_T(0)\|.
\]
Combining these estimates gives the result.
\end{proof}

\begin{remark}
The projection \(\hat P_j(s)\) need not be orthogonal with respect to the
standard inner product. Therefore,
\(
\|(I-\hat P_j(s))x\|
\)
measures the component selected by the complementary projection and should not,
in general, be interpreted as the orthogonal distance from \(x\) to
\(\operatorname{Ran}\hat P_j(s)\).
\end{remark}

The preceding corollary gives a quantitative leakage estimate in the original
non-Hermitian representation. It controls the component selected by the
complementary projection and therefore gives a statement directly in terms of
the genuine spectral subspaces of the original Hamiltonian. The next section
illustrates the hypotheses and the resulting estimate through a two-level
non-Hermitian example.

\section{Example: A nontrivial two-level non-Hermitian model}

In this section we present a two-level example illustrating the full structure of
the main theorem. The original Hamiltonian is non-Hermitian, the Dyson map is
time-dependent, and the Hermitian counterpart is nonconstant. Consequently, both
the pulled-back Hermitian adiabatic error and the projection-comparison error
appear in the final estimate.

Let,
\[
\sigma_x=
\begin{pmatrix}
0&1\\
1&0
\end{pmatrix},
\qquad
\sigma_y=
\begin{pmatrix}
0&-i\\
i&0
\end{pmatrix},
\qquad
\sigma_z=
\begin{pmatrix}
1&0\\
0&-1
\end{pmatrix}.
\]
Let \(\Delta>0\), \(a>0\), and \(\Omega_0>0\). Define,
\[
r(s):=a\sin(\pi s),
\qquad
\Omega(s):=\Omega_0\sin(2\pi s),
\qquad s\in[0,1].
\]
Consider the non-Hermitian Hamiltonian,
\[
\begin{aligned}
\hat H_T(s)
&=
-\frac{i}{T}r'(s)\sigma_x +
\frac{1}{2}
\left(
\Omega(s)\cosh(2r(s))
+
i\Delta\sinh(2r(s))
\right)\sigma_y  \\
&\quad+
\frac{1}{2}
\left(
\Delta\cosh(2r(s))
-
i\Omega(s)\sinh(2r(s))
\right)\sigma_z .
\end{aligned}
\]
This Hamiltonian is generally non-Hermitian, because the coefficients of the
Hermitian matrices \(\sigma_y\), \(\sigma_z\), and \(\sigma_x\) are not all real.
Define,
\[
\hat\eta_T(s):=e^{r(s)\sigma_x}.
\]
Since \(\sigma_x\) is Hermitian and \(r(s)\) is real, \(\hat\eta_T(s)\) is
positive and invertible for every \(s\in[0,1]\). Moreover,
\[
\hat\eta_T(s)^{-1}=e^{-r(s)\sigma_x}.
\]
Using \(\sigma_x^2=I\), we have,
\[
e^{r(s)\sigma_x}
=
\cosh(r(s))I+\sinh(r(s))\sigma_x.
\]
Since \(\sigma_x\) commutes with \(e^{r(s)\sigma_x}\), it follows that,
\[
\dot{\hat\eta}_T(s)\hat\eta_T(s)^{-1}
=
r'(s)\sigma_x.
\]
Therefore the Dyson correction term is,
\[
B_T(s)
=
\frac{i}{T}\dot{\hat\eta}_T(s)\hat\eta_T(s)^{-1}
=
\frac{i}{T}r'(s)\sigma_x.
\]
Since \(\|\sigma_x\|=1\), we obtain,
\[
\|B_T(s)\|
=
\frac{|r'(s)|}{T}.
\]
For \(r(s)=a\sin(\pi s)\),
\[
r'(s)=a\pi\cos(\pi s),
\]
and hence,
\[
\|B_T(s)\|
=
\frac{a\pi|\cos(\pi s)|}{T}
\le
\frac{a\pi}{T}.
\]
Thus the Dyson correction bound holds with,
\[
b_T(s)=|r'(s)|,
\qquad
\beta_{T}=a\pi.
\]
We now compute the Dyson transform of \(\hat H_T(s)\). From the Pauli matrix
relations,
\[
e^{r(s)\sigma_x}\sigma_y e^{-r(s)\sigma_x}
=
\cosh(2r(s))\sigma_y
+
i\sinh(2r(s))\sigma_z
\]
and,
\[
e^{r(s)\sigma_x}\sigma_z e^{-r(s)\sigma_x}
=
\cosh(2r(s))\sigma_z
-
i\sinh(2r(s))\sigma_y.
\]
Also,
\[
e^{r(s)\sigma_x}\sigma_x e^{-r(s)\sigma_x}
=
\sigma_x.
\]
A direct substitution gives,
\[
\hat\eta_T(s)\hat H_T(s)\hat\eta_T(s)^{-1}
=
\frac{\Delta}{2}\sigma_z
+
\frac{\Omega(s)}{2}\sigma_y
-
\frac{i}{T}r'(s)\sigma_x.
\]
Adding the Dyson correction term, we obtain,
\[
\begin{aligned}
\bar H_T(s)
&=
\hat\eta_T(s)\hat H_T(s)\hat\eta_T(s)^{-1}
+
\frac{i}{T}\dot{\hat\eta}_T(s)\hat\eta_T(s)^{-1} \\
&=
\left(
\frac{\Delta}{2}\sigma_z
+
\frac{\Omega(s)}{2}\sigma_y
-
\frac{i}{T}r'(s)\sigma_x
\right)
+
\frac{i}{T}r'(s)\sigma_x \\
&=
\frac{\Delta}{2}\sigma_z
+
\frac{\Omega(s)}{2}\sigma_y.
\end{aligned}
\]
Hence the Dyson transform of \(\hat H_T(s)\) is the Hermitian Hamiltonian,
\[
\bar H_T(s)
=
\frac{\Delta}{2}\sigma_z
+
\frac{\Omega(s)}{2}\sigma_y.
\]
Since \(\bar H_T(s)\) is Hermitian, the metric,
\[
\hat\Theta_T(s):=\hat\eta_T(s)^*\hat\eta_T(s)
\]
satisfies the compatibility relation,
\[
\hat H_T(s)^*\hat\Theta_T(s)-\hat\Theta_T(s)\hat H_T(s)
=
\frac{i}{T}\dot{\hat\Theta}_T(s).
\]
Thus the chosen Dyson map is dynamically compatible with the non-Hermitian
evolution considered in the example.

We next verify the Hermitian adiabatic assumption for \(\bar H_T(s)\). Since,
\[
\bar H_T(s)
=
\frac12
\begin{pmatrix}
\Delta&-i\Omega(s)\\
i\Omega(s)&-\Delta
\end{pmatrix},
\]
we have,
\[
\det(\bar H_T(s)-\lambda I)
=
\lambda^2-\frac{\Delta^2+\Omega(s)^2}{4}.
\]
Thus the eigenvalues are,
\[
\bar E_{\pm}(s)
=
\pm \frac{1}{2}R(s),
\qquad
R(s):=\sqrt{\Delta^2+\Omega(s)^2}.
\]
Therefore the spectral gap is,
\[
\bar g(s)
=
|\bar E_+(s)-\bar E_-(s)|
=
R(s).
\]
Since \(\Delta>0\), we have,
\[
R(s)=\sqrt{\Delta^2+\Omega(s)^2}\ge\Delta>0.
\]
Hence the eigenvalues of \(\bar H_T(s)\) remain separated on
\([0,1]\).
Moreover,
\[
\dot{\bar H}_T(s)=\frac{\Omega'(s)}{2}\sigma_y,
\qquad
\ddot{\bar H}_T(s)=\frac{\Omega''(s)}{2}\sigma_y.
\]
Since,
\[
\Omega'(s)=2\pi\Omega_0\cos(2\pi s),
\qquad
\Omega''(s)=-4\pi^2\Omega_0\sin(2\pi s),
\]
the Hermitian adiabatic constant is generally nonzero. In this example,
\[
\bar C_{\pm}(s)
=
\frac{|\Omega'(s)|}{2R(s)^2}
+
\frac{|\Omega'(0)|}{2R(0)^2}
+
\int_0^s
\left(
\frac{|\Omega''(u)|}{2R(u)^2}
+
\frac{5}{2}\frac{|\Omega'(u)|^2}{R(u)^3}
\right)du.
\]
Therefore the pulled-back Hermitian adiabatic contribution,
\[
\frac{\bar C_{\pm}(s)}{T}
\]
is nontrivial.

We now check the spectral admissibility of the original non-Hermitian
Hamiltonian. Since,
\[
\hat\eta_T(s)\hat H_T(s)\hat\eta_T(s)^{-1}
=
\bar H_T(s)-B_T(s),
\]
the Hamiltonian \(\hat H_T(s)\) is similar to,
\[
\bar H_T(s)-B_T(s)
=
\frac{\Delta}{2}\sigma_z
+
\frac{\Omega(s)}{2}\sigma_y
-
\frac{i}{T}r'(s)\sigma_x.
\]
Hence \(\hat H_T(s)\) and \(\bar H_T(s)-B_T(s)\) have the same eigenvalues.
Writing,
\[
M_T(s):=\bar H_T(s)-B_T(s),
\]
we obtain,
\[
M_T(s)
=
\begin{pmatrix}
\dfrac{\Delta}{2}
&
-\dfrac{i\Omega(s)}{2}-\dfrac{i r'(s)}{T}
\\[6pt]
\dfrac{i\Omega(s)}{2}-\dfrac{i r'(s)}{T}
&
-\dfrac{\Delta}{2}
\end{pmatrix}.
\]
A direct computation gives,
\[
\det(M_T(s)-\lambda I)
=
\lambda^2
-
\frac{\Delta^2+\Omega(s)^2}{4}
+
\frac{(r'(s))^2}{T^2}.
\]
Therefore the eigenvalues of \(\hat H_T(s)\) are,
\[
\lambda_{\pm}(s)
=
\pm
\frac{1}{2}
\sqrt{
\Delta^2+\Omega(s)^2-\frac{4(r'(s))^2}{T^2}
}.
\]
These eigenvalues are real and distinct whenever,
\[
\Delta^2+\Omega(s)^2-\frac{4(r'(s))^2}{T^2}>0,
\qquad s\in[0,1].
\]
Since \(\Delta^2+\Omega(s)^2\ge\Delta^2\) and,
\[
|r'(s)|\le a\pi,
\]
it is enough to assume,
\[
T>\frac{2a\pi}{\Delta}.
\]
Under this condition, the two eigenvalues of \(\hat H_T(s)\) are real,
distinct, and non-crossing. Since \(\hat H_T(s)\) is a \(2\times2\) matrix with
two distinct eigenvalues, it is diagonalizable. Hence \(\hat H_T(s)\) satisfies
Assumption~\ref{Ass:NonHermitianAdiabatic}.

It remains to discuss the contour-resolvent condition. Define,
\[
A_T(s):=\hat\eta_T(s)\hat H_T(s)\hat\eta_T(s)^{-1}
=
\bar H_T(s)-B_T(s).
\]
Since,
\[
\|B_T(s)\|\le \frac{a\pi}{T},
\]
we have \(A_T(s)\to \bar H_T(s)\) uniformly on \([0,1]\) as \(T\to\infty\).
The Hermitian operator \(\bar H_T(s)\) has spectral gap bounded below by
\(\Delta>0\). Hence, for sufficiently large \(T\), the two eigenvalues
of \(A_T(s)\) remain uniformly separated. Therefore one may choose contours
\(\Gamma_{\pm}(s)\) enclosing the corresponding eigenvalue of \(A_T(s)\) and no
other point of \(\sigma(A_T(s))\), with resolvent bounds satisfying,
\[
M_{\pm}(s)
:=
\sup_{z\in\Gamma_{\pm}(s)}
\|(zI-A_T(s))^{-1}\|
\le C_{\pm}
\]
for some constant \(C_{\pm}\) independent of \(T\), for all sufficiently large
\(T\). Consequently,
\[
M_{\pm}(s)\|B_T(s)\|
\le
C_{\pm}\frac{a\pi}{T}.
\]
Thus, for sufficiently large \(T\),
\[
M_{\pm}(s)\|B_T(s)\|<1.
\]

Therefore, for sufficiently large \(T\), all hypotheses of
Theorem~\ref{Thm:MainNonHermitianEigen} are satisfied. Hence, for every initial
state,
\[
|\psi(0)\rangle\in\operatorname{Ran}\hat P_{\pm}(0),
\]
the theorem gives,
\[
\begin{aligned}
\left\|
\bigl(\widetilde U_{A,\pm}(s)-\hat U_T(s)\bigr)|\psi(0)\rangle
\right\|
&\le
\left[
\|\hat\eta_T(s)^{-1}\|\,\|\hat\eta_T(0)\|
\frac{\bar C_{\pm}(s)}{T}
\right.\\
&\quad\left.
+
2\|\hat\eta_T(s)^{-1}\|\,\|\hat\eta_T(0)\|
\varepsilon_{\pm}(0)
\right]
\||\psi(0)\rangle\|.
\end{aligned}
\]
Since \(r(0)=0\), we have,
\[
\hat\eta_T(0)=I.
\]
Also, since the eigenvalues of \(\hat\eta_T(s)^{-1}=e^{-r(s)\sigma_x}\) are
\(e^{-r(s)}\) and \(e^{r(s)}\), we have,
\[
\|\hat\eta_T(s)^{-1}\|=e^{|r(s)|}.
\]
Therefore the estimate becomes,
\[
\left\|
\bigl(\widetilde U_{A,\pm}(s)-\hat U_T(s)\bigr)|\psi(0)\rangle
\right\|
\le
e^{|r(s)|}
\left[
\frac{\bar C_{\pm}(s)}{T}
+
2\varepsilon_{\pm}(0)
\right]
\||\psi(0)\rangle\|.
\]

This example illustrates the full structure of the non-Hermitian estimate. The
Hermitian counterpart \(\bar H_T(s)\) is time-dependent, so the Hermitian
adiabatic contribution \(\bar C_{\pm}(s)/T\) is generally nonzero. The Dyson map
is also time-dependent, so the Dyson correction \(B_T(s)\) is nonzero and gives
rise to the projection-comparison contribution \(\varepsilon_{\pm}(0)\). Hence
both terms appearing in the main theorem are present in this example.

\section{Discussion and Conclusion}

The main result of this paper is a quantitative adiabatic estimate for a class of
finite-dimensional non-Hermitian Schr\"odinger dynamics. Starting from the exact
evolution, a dynamically compatible metric operator is constructed, and its
positive square root defines a time-dependent Dyson map. This Dyson
transformation yields a Hermitian counterpart of the original non-Hermitian
system whenever the transformed Hamiltonian satisfies the Hermitian adiabatic
assumption.

A key aspect of the analysis is that the Dyson-pulled-back projections are not,
in general, identical to the spectral projections of the original non-Hermitian
Hamiltonian. To address this issue, the Hermitian counterpart is expressed as a
perturbation of a similarity transform of the original Hamiltonian. A contour-resolvent estimate is then used to compare the corresponding projection
families. This comparison is responsible for the second term in the final estimate, while
the first term comes from the Hermitian adiabatic theorem after Dyson pullback.

The two-level example in Section~5 illustrates how these assumptions can be
verified in a concrete finite-dimensional model. In that example, the Hermitian counterpart is non-constant and the Dyson correction is nonzero, so both the
pulled-back Hermitian adiabatic term and the projection-comparison term appear
explicitly.

The present framework has several limitations. First, the analysis is restricted
to finite-dimensional systems and does not address questions involving unbounded
operators, domain issues, or infinite-dimensional spectral theory. Second, the
metric operator is constructed from the exact evolution itself. While this
establishes the existence of a dynamically compatible metric and Dyson map, it
does not provide an independent constructive procedure for determining them in
concrete models. Consequently, the applicability of the theorem depends on
verifying the Hermitian adiabatic assumption for the Dyson-transformed
Hamiltonian and the contour-resolvent condition needed for comparing spectral
projections.

There are several natural directions for future investigation. One is the
extension of the present theory to infinite-dimensional settings. Another is the
development of constructive criteria for obtaining controlled Dyson maps without
prior knowledge of the exact evolution. It would also be desirable to identify
broad classes of non-Hermitian Hamiltonians for which the Hermitian adiabatic
assumption and the contour-resolvent estimates can be verified explicitly.

In summary, the paper shows that a quantitative adiabatic estimate can be
transferred from a Hermitian representation to a class of non-Hermitian dynamics
and then reformulated in terms of the spectral projections of the original
Hamiltonian. The resulting estimate is expressed directly in the original
non-Hermitian spectral subspaces, rather than only in their Dyson-transformed
counterparts.

\bibliography{sn-bibliography}

\end{document}